\documentclass[10pt, conference]{IEEETran}

\usepackage{microtype}
\usepackage{graphicx}
\usepackage{pifont}
\usepackage{booktabs} % for professional tables
\usepackage{caption}
\usepackage{subcaption}
\usepackage{bm}
\usepackage{enumitem}
\usepackage[mathlines,switch]{lineno}

\usepackage{changes}
\definechangesauthor[name="Ilya Safro", color=red]{is}
\definechangesauthor[name="Ankit", color=purple]{ak}
\definechangesauthor[name="Joey", color=blue]{jl}
\definechangesauthor[name="Hayato", color=blue]{hm}

\usepackage{hyperref}

\usepackage{amsmath}
\usepackage{amssymb}
\usepackage{mathtools}
\usepackage{amsthm}
\usepackage[capitalize,noabbrev]{cleveref}

\theoremstyle{plain}
\newtheorem{theorem}{Theorem}[section]

\theoremstyle{definition}

\theoremstyle{remark}

\newcommand{\ket}[1]{\ensuremath{|#1 \rangle}}
\newcommand{\bra}[1]{\ensuremath{\langle #1 |}}

\newcommand{\U}{\bm{U}(\bm{\theta})}
\newcommand{\Uad}{\bm{U}^{\dagger}(\bm{\theta})}
\newcommand{\fQ}{f_{\theta}}

\title{Learning To Optimize Quantum Neural Network Without Gradients}

\author{\IEEEauthorblockN{Ankit Kulshrestha}
\IEEEauthorblockA{\textit{Department of Computer and Information Sciences}) \\
\textit{University of Delaware}\\
Newark, USA\\
akulshr@udel.edu}
\and
\IEEEauthorblockN{Xiaoyuan Liu}
\IEEEauthorblockA{
\textit{Fujitsu Research of America}\\
Sunnyvale, USA \\
xliu@fujitsu.com}
\and
\IEEEauthorblockN{Hayato Ushijima-Mwesigwa}
\IEEEauthorblockA{
\textit{Fujitsu Research of America}\\
Sunnyvale, USA \\
hayato@fujitsu.com}
\and
\IEEEauthorblockN{Ilya Safro}
\IEEEauthorblockA{\textit{Department of Computer and Information Sciences} \\
\textit{University of Delaware}\\
Newark, US \\
isafro@udel.edu}}

\begin{document}
\maketitle
\begin{abstract}
Quantum Machine Learning is an emerging sub-field in machine learning where one of the goals is to perform pattern recognition tasks by encoding data into quantum states. This extension from classical to quantum domain has been made possible due to the development of hybrid quantum-classical algorithms that allow a parameterized quantum circuit to be optimized using gradient based algorithms that run on a classical computer. The similarities in training of these hybrid algorithms and classical neural networks has further led to the development of Quantum Neural Networks (QNNs). 
However, in the current training regime for QNNs, the gradients w.r.t objective function have to be computed on the quantum device. This computation is highly non-scalable and is affected by hardware and sampling noise present in the current generation of quantum hardware. In this paper, we propose a training algorithm that does not rely on gradient information. Specifically, we introduce a novel meta-optimization algorithm that trains a \emph{meta-optimizer} network to output parameters for the quantum circuit such that the objective function is minimized. We empirically and theoretically show that we achieve a better quality minima in fewer circuit evaluations than existing gradient based algorithms on different datasets.\\
%The source code, data and results will be available upon acceptance of the paper.
\end{abstract}

\section{Introduction}
\label{sec:introduction}
Machine learning has evolved over time from solving small size pattern recognition problems to being able to capture latent structure in data and generalize to unseen data at scale. The state of machine learning is quickly approaching a point where classical computers would not be able to keep up with the ever increasing scale and complexity of data. On the other hand, quantum computing holds a theoretical promise of being able to scale beyond existing classical approaches. While the current state of quantum machines is still far from being practically competitive to the classical counterparts, the emergent field of Quantum Machine Learning (QML) is an important step towards the future in which quantum computers form the basis of many challenging computational tasks \cite{herman2022survey,ajagekar2019quantum}.

The current bridge between classical and quantum algorithms is a class of algorithms called Variational Quantum Algorithms (VQA)s~\cite{cerezo2021variational}. These algorithms are concerned with finding the optimal set of parameters for a Variational Quantum Circuit (VQC) to minimize an objective function. The parameters are optimized using gradient descent which runs on a classical computer. Since the training procedure is very similar to modern Deep Neural Networks (DNNs), quantum variants of several existing DNN architectures have been proposed~\cite{farhi2018classification, henderson2020quanvolutional, romero2017quantum, cong2019quantum}. The classical and quantum neural networks differ in one key aspect: the  data has to be encoded as quantum states before it can be processed by a QNN. This encoding can be seen as a mapping to a high dimensional Hilbert space where the data is separable. It is thus entirely possible that with a powerful quantum computer, we will be able to capture a much richer representation of data and consequently outperform DNNs which is the motivation and hope behind several QML algorithms. The VQCs are deployed on quantum devices that support operations on quantum states. The current generation of devices, called Noisy Intermediate Scale Quantum (NISQ) devices, is quite limited due to a high presence of gate noise, inability to scale to beyond a few hundred qubits, and lack of  reliable error mitigation and detection algorithms~\cite{cerezo2021variational}. These bottlenecks directly affect the performance of QNNs and hence motivate  better algorithms to train them.

% The ongoing development of NISQ computers has led to an extensive research in algorithms that can utilize them and get closer to demonstrating a practical quantum advantage over classical methods. 
% Hybrid quantum-classical algorithms are considered to be one of the most prominent ways of achieving it in the observable future. 
% The classical computer is used to find the best parameters for a quantum circuit that will subsequently be executed on a quantum device. Then, the result of execution on a quantum device is returned back to the classical machine to improve the parametrization and the same loop is repeated \cite{liu2022layer,cerezo2021variational}.

% Given the diverse applications of VQAs despite their early stage of development gives hope that they will eventually be able to demonstrate a quantum advantage over classical algorithms by either finding a more optimal solution or converging to a solution in a polynomial time for a problem size that is intractable by classical methods. However, the current generation of NISQ devices have several limitations~\cite{cerezo2021variational} including but not limited to having shallow depth circuit and a small number of qubits to encode a given problem. Moreover, a quantum error detection and correction mechanism is still being actively researched and has not reached the necessary maturity to reliably detect and correct errors in quantum circuits.   
% \added[id=is]{We are in the end of second paragraph and there is nothing yet about ML; need inject something}

\emph{\textbf{The need for gradient-free optimization algorithms:}}  The bottlenecks in current generation of quantum devices are not the only reasons that motivate a better optimization algorithm. Other factors also make the case for a better optimization algorithm more compelling. For instance, for a QNN running on a quantum device, the gradients are computed using the parameter shift rule~\cite{schuld2019evaluating}. This method scales as $O(N)$ where $N$ is the number of parameters of the QNN. A full forward-backward pass then scales as $O(N^2)$. Clearly, if the QNN is to be scaled to large problem instances (e.g., Imagenet~\cite{deng2009imagenet} classification) the quadratic cost of computation must be improved. 

Another contributing reason is that in the current training regime, the QNN is treated as a black-box from the perspective of the optimizer. While conventional optimization algorithms like RMSProp~\cite{tieleman2012lecture}, Adam~\cite{kingma2014adam}, SGD etc.  still work for QNNs, it is possible that when QNNs are scaled to large problem instances, hand-designed optimization rules may not be able to fully capture the complexities of the probabilistic nature of QNNs. These reasons, coupled with inherent issues in NISQ devices  motivate the development of a efficient learned optimizers. More crucially, in order to be broadly applicable we demand that the optimizer makes as few calls to the quantum circuit as possible and estimate the direction of optimization in a gradient free manner. 

% TODO: revise language of the meta-optimizer.
One way of designing a new learning rule is to \emph{learn} it during training for a fixed task and dataset. If $\bm{\theta}^t \in \mathbb{R}^{N}$ are the parameters of QNN at timestep $t$ and $C(\bm{\theta})$ is a cost function we are interested in minimizing, then a learned update rule is of the form $\bm{\theta}^{t+1} = \mathcal{R}_{\bm{\Phi}}(\bm{\theta}^t, \nabla_{\theta} C(\bm{\theta}^{t}))$ where $\nabla_{\theta} C(\bm{\theta}^{t})$ is the gradient of QNN cost w.r.t $\bm{\theta}^{t}$. Here, $\mathcal{R}_{\bm{\Phi}}$ is a DNN parameterized by $\bm{\Phi} \in \mathbb{R}^{M}$ where $M$ is the number of parameters in the DNN. $\mathcal{R}_{\bm{\Phi}}$ is trained using a meta-loss function $\mathcal{L}(\bm{\Phi})$. For any given dataset and QNN, we are interested in finding an optimal set of meta-parameters $\bm{\Phi}^*$ by minimizing $\mathcal{L}(\bm{\Phi})$ such that the QNN cost $C(\bm{\theta})$ is minimized. For the rest of the paper, we refer to $\mathcal{R}_{\bm{\Phi}}$ as a \emph{meta-optimizer} and the problem of minimizing $\mathcal{L}(\bm{\Phi})$ as 
meta-optimization.

% One way of designing a new learning rule is to \emph{learn} it from the data. More formally, let $\bm{\theta}$ be a set of $n$ parameters, and $C(\bm{\theta})$ be an objective function we are minimizing. If at time step $t$, the parameters are $\bm{\theta}^{t}$ and the corresponding value of objective function is $C(\bm{\theta}^{t})$  then we define a learned update rule to be 
% \begin{equation}
% \bm{\theta}^{t+1} = \mathcal{L}(\bm{\theta}^{t}, f(C(\bm{\theta}^{t}), \bm{\Phi})), 
% \end{equation}
% such that $C(\bm{\theta}^{t+1})<C(\bm{\theta}^{t})$. $\mathcal{L}: \mathbb{R}^{n} \times \mathbb{R} \mapsto \mathbb{R}^n$ is a non-linear function that accepts $(\bm{\theta}^{t}, f(C(\bm{\theta}^{t}), \bm{\Phi}))$ as inputs and $f$ is a map from input data to output labels and adjusts its own parameters $\bm{\Phi}$ to minimize a \emph{meta-loss} function. We refer to the problem of finding optimal meta-parameters $\bm{\Phi}^{*}$ as \emph{meta-optimization} and $\mathcal{L}$ as a \emph{meta-optimizer}. 

Meta-optimization approaches are well studied in the deep learning literature. Andrychowicz~\emph{et al.}~\cite{andrychowicz2016learning} introduced the idea of optimizing a deep neural network using an LSTM based meta-optimizer that accepted $(\bm{\theta}^{t}, \nabla_{\theta} C(\bm{\theta}^{t}))$ as inputs. Li and Malik~\cite{li2016learning, li2017learning}  explore the idea from the reinforcement learning perspective where they leverage the LSTM based optimizer to learn a \emph{policy} for predicting the next set of parameters. Metz~\emph{et al.}~\cite{metz2019understanding} propose an alternative algorithm for fast training of meta-optimizers that  generalizes better than the gradient descent training. 

These findings have been inherited for training variational quantum circuits (VQCs) in the works of Wilson~\emph{et al.}~\cite{wilson2021optimizing} and Verdon~\emph{et al.}~\cite{verdon2019learning}. In the former work, the authors build on the algorithm in~\cite{andrychowicz2016learning} for a VQC and leverage $\nabla_{\theta} C(\bm{\theta})$ as an input feature for the meta-optimizer.  In the latter work, the authors propose using the meta-optimization framework for learning initial parameters using a recurrent neural network and then initializing a VQC with these learned parameters. The VQC is then proposed to be fine-tuned by a regular optimizer until a desired accuracy level is reached. In all the aforementioned works, a common denominator  is the reliance over the gradient as input features to the meta-optimizer (except ~\cite{verdon2019learning} which uses  $C(\bm{\theta})$). Given that computing gradients for QNNs is not cheap, the challenge is to develop a meta optimizer that can learn a proper update rule \emph{without} any explicit gradient computation step. 
% \emph{The challenge in the quantum domain is to design a strategy that performs competitively to classical optimizers without involving an explicit quantum gradient computation step which is what we propose}.

\noindent{\textbf{Our Contribution}} 
We address this critical challenge and propose a novel meta-optimization algorithm that learns to train QNNs without relying on any gradient information. More specifically, we make the following contributions:

% \added[id=is]{can we say that we propose the first gradient-free strategy?} \added[id=jl]{I don't think so, since there are other black-box optimization algorithms such as COBYLA that are gradient-free.} In brief, we make the following contributions:

%\begin{itemize}
\noindent (1) We design a novel algorithm for training QNNs using a meta-optimizer in a way that does not involve any computation or approximation of $\nabla_{\theta} C(\bm{\theta})$ on a quantum device. We show that using better features and training schedule can result in a meta-optimizer which is competitive with gradient based optimizers that are currently used for training QNNs.
% \noindent (1) We introduce a novel fast meta-optimization algorithm for training VQCs that does not rely on any gradient computation on the quantum processor. We show that engineering better features and training updates can help in simulating gradient descent effectively.
    
\noindent (2) We theoretically and empirically verify that our algorithm provides a significant speedup in training time over conventional gradient based optimization algorithms. 
    
\noindent (3) We empirically demonstrate that our algorithm achieves a minima which is comparable to that attained by the conventional first-order methods. Moreover, we show that with right initialization our algorithm can outperform classical gradient based algorithms with the same initialization.

To the best of our knowledge, this is a first meta-optimization algorithm for training QNNs that does not rely on computing gradients. We advocate that our algorithm can serve as blueprint for gradient-free meta-optimization algorithms.
% To the best of our knowledge, In summary, we advocate that our algorithm can serve as a blueprint for gradient-free optimization algorithms for training VQCs in the future.
%\end{itemize}

% TODO: write breadcrumbs for readers.

% The rest of the paper is organized as follows. In Section~\ref{sec:}
% interpreted to be a non-linear combination of inputs that learns the gradient trajectory \emph{without} explicitly computing the gradient.
% We shall make the notion of $\mathcal{L}$ more concrete in the coming sections. F

% solve the optimization problem using a \emph{Variational Quantum Circuit}(VQC) which is deployed on a quantum device and consists of parameterized unitary gates (generally drawn from the Clifford group) arranged in a layer-wise manner. The parameters $\bm{\theta}$ are real-valued parameters. The VQC is expected to minimize a cost function $C(\bm{\theta})$ defined as:
%\section{Background}
%\label{sec:theory}
%In this section we introduce key ideas and the overall algorithm that form the core of our work.

%% better to color code arrows _to_ and _from_ LSTM for better understanding
\begin{figure*}[t]
    \centering
    \includegraphics[width=\textwidth]{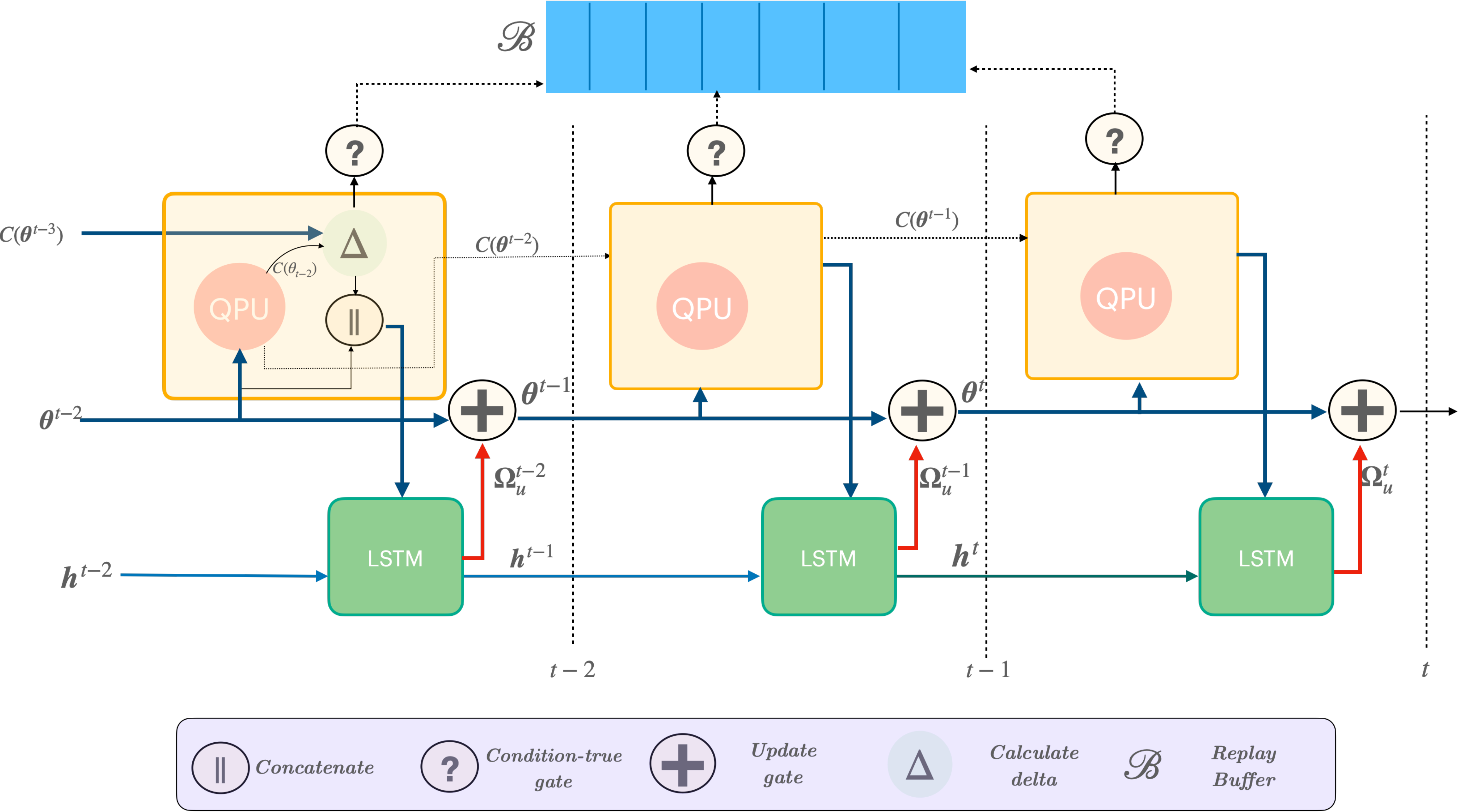}
    \caption{An overview of our proposed meta-learning algorithm. The blue arrows indicate key inputs to the QNN and the LSTM and the red
    arrows indicate the output from LSTM. Best viewed in color. See Section~\ref{sec:algorithm} for more details.}
    \label{fig:rnn_flow}
\end{figure*}
% \added[id=hm]{Need to explicitly say the difference between bold theta and the other. Should be the product of unitary \textbf{matrices}? In general, I would use Uppercase bold for Matrices and lowercase bold for vectors. So perhaps $\bm{U}_{\theta_i}$ and $\bm{U}_{\bm{\theta}}$? It is not clear if $U_e$ is a vector or matrix.} 
\section{Quantum Neural Networks}
For completeness, we first introduce the idea of Quantum Neural Networks (QNNs)~\cite{cerezo2019variational, farhi2018classification}.  Given a dataset consisting of $m$ samples $\mathcal{D} = \{\bm{x}_{i}, y_{i}\}_{i=1}^{m}$, where $\bm{x}_i \in \mathbb{R}^{d}$ are the data points with  corresponding labels $y_i$\footnote{The index subscript will be omitted where it is clear from the context.}. A quantum neural network $\fQ: \mathbb{R}^{d} \mapsto \mathbb{R}$ is a variational quantum circuit that learns to assign $y$ to $\bm{x}\in \mathbb{R}^d$. A parameterized quantum circuit implementing a quantum neural network can be written as a product of $n$  parametrized unitary matrices $\bm{U}_i(\theta_i)$, $\bm{U}(\bm{\theta}) = \prod_{i=1}^{n} \bm{U}_{i}(\theta_{i})$. Before the data can be processed by the VQC, an encoding circuit is applied to prepare the input state $\ket{\psi} = \bm{U}_0(\bm{x})\ket{0}$, where $\bm{U}_0(\bm{x}) = e^{-i\bm{x}\bm{G}}$ is a unitary  that encodes each component of a single $d$ dimensional data vector into a quantum state consisting of $q = \log(d)$ qubits. Here, $\bm{G}$ is a gate generating Hermitian matrix.  Overall we can describe the QNN $\fQ(\bm{x}, \bm{\theta})$ as:
\begin{equation}
    \fQ(\bm{x}, \bm{\theta}) = \bra{0} \bm{U}^{\dagger}_{0}(\bm{x})\Uad \bm{\hat{O}} \U \bm{U}_{0}(\bm{x})\ket{0},
    \label{eq:qnn}
\end{equation}
where $\bm{\hat{O}}$ is a quantum observable that maps the output quantum state to a scalar number and $\bm{U}^{\dagger}$ are the complex conjugate of the unitary matrix $\bm{U}$. We select the minimum squared error as the cost function of our choice. Averaged over all $m$ points of the dataset, we express it as:
\begin{equation}
    C(\bm{\theta}) = \frac{1}{m} \sum_{i=1}^{m} ||y_{i} - \fQ(\bm{x}_{i}, \bm{\theta})||^{2}_{2}
    \label{eq:costfn}
\end{equation}
% TODO: Change the argmin properly
The objective is to find $\bm{\theta}^{*} = \operatorname{argmin}_{\theta} C(\bm{\theta})$. In this work we shall assume that the QNN is minimizing the cost function, $C(\bm{\theta)}$, in Equation~(\ref{eq:costfn}).

\section{Meta-Optimization Framework}
We now discuss the meta-optimization framework in more detail. For simplicity, we adopt a similar nomenclature as in~\cite{andrychowicz2016learning} and refer to the QNN as the \emph{optimizee} network and the meta-optimizer as the \emph{optimizer} network.

In the meta-optimization framework, the task of the optimizee network is to evaluate the parameters suggested by the optimizer network. In turn, the optimization network accepts some input meta-features and performs an internal computation to generate new parameters for the optimizee network while adjusting its own parameters. In particular, if $\bm{\theta}^{t}$ are the optimizee parameters at time-step $t$ during training, then the next parameters are discovered following~\cite{andrychowicz2016learning}:
\begin{equation}
 \begin{aligned}[t]
    \bm{\theta}^{t+1} &= \bm{\theta}^{t} + \bm{g}^t\,,
    \\
    \begin{bmatrix}
     \bm{g}^{t} \\ \bm{h}^{t+1} 
    \end{bmatrix}
    &= \mathcal{R}_{\bm{\Phi}}(\mathcal{I}(\bm{\theta}^{t}), \bm{h}^{t})\,
 \end{aligned}.
 \label{eq:optimizee-adjust}
\end{equation}
% check the claim about I(\theta) for the Li paper.
Here we specify  $\mathcal{R}_{\bm{\Phi}}$ to be an instance of LSTM that accepts a hidden state $\bm{h}^{t}$ and the meta-parameters $\bm{\Phi}$ and outputs the update $\bm{g}^{t}$. The LSTM accepts $\mathcal{I}(\bm{\theta}^{t})$ which we define to be a function that combines information from the training process and presents it as input features to the optimizer network. For instance, in~\cite{andrychowicz2016learning} and similar works~\cite{li2017learning, wilson2021optimizing}, $\mathcal{I}(\bm{\theta}) = \nabla_{\theta} C(\bm{\theta})$. The general form of meta-loss function tries to minimize the loss over a finite horizon window of size $T$. In this work, we follow a meta-loss function inspired by the deep learning literature:
\begin{equation}
    \mathcal{L}(\bm{\Phi}) = \sum_{i=1}^{T} w_t C(\bm{\theta}^t),
    \label{eq:meta-optimizer}
\end{equation}
where $\bm{W}^{m} = \left[w_{1}, w_{2}, \dots, w_{T} \right]$ are the weights over the cost function evaluations at time step $t$. A uniform weighting strategy sets all weights to be equal. The \emph{magnitude} of the weights is a matter of choice and the decision is made based on the dataset that is being used. The optimization of $\bm{w}$ in a data driven manner is out of scope of this paper. In general, we suggest that a line of work based on hyper-parameter optimization~\cite{bergstra2011algorithms} can be developed and used to dynamically adjust the weights for any given data.

\section{Our Algorithm}
\label{sec:algorithm}
We will now discuss the main components of our algorithm. Although, we generally follow the meta-optimization framework discussed above, the constraints induced by the need for gradient free optimization lead us to introduce the algorithmic tools that have not been proposed in literature in the meta-optimization context for quantum machine learning.

% maybe add the sign that's part of the input?
\textbf{Input Preprocessing}: In earlier works, the input function involved passing some gradient information from the optimizee network to the optimizer network. Since we are constrained to not use gradient information, we utilize the tuple $(\bm{\theta}^{t}, \Delta C(\bm{\theta})$) as inputs where $\Delta C(\bm{\theta})$ is a \emph{pseudo-gradient} that computes the difference between the current and previous cost function evaluation, i.e., $\Delta C(\bm{\theta}) = C(\bm{\theta}^{t-1}) - C(\bm{\theta}^{t})$. These inputs are concatenated in a single vector and the input function $\mathcal{I}(\bm{\theta})$ is prepared as:
\begin{equation}
    \mathcal{I}(\bm{\theta})  = \frac{\mathcal{P}([\bm{\theta}^t; \Delta C(\bm{\theta}^t)])}{p}.
    \label{eq:input-function}
\end{equation}
Here we use a non-linear function $\mathcal{P}: \mathbb{R}^{d} \mapsto [0, 1]$ that normalizes the input values to between 0 and 1. The value $p$ controls the strength of normalization. In our experiments we use $\mathcal{P}(\bm{x}) = e^{\bm{x}}$ with $\bm{x}$ being the input vector for the LSTM. We empirically found $p=50$ to work best for the datasets considered in this study. In a future work, we would like to explore the possibility of learning an adaptive normalization strength that is derived from observing history of updates during training.

\textbf{Non Linear Parameter Updates}: In Equation~(\ref{eq:optimizee-adjust}), we defined a generic update rule in the meta-optimization framework. That update rule simply adds the output of LSTM to previously obtained parameters. In this work, we introduce a new update rule of the form:
\begin{equation}
    \bm{\theta}^{t+1} = \bm{\theta}^{t} + \alpha \cdot \sigma(\bm{\Omega}^{t}_{u}),
    \label{eq:param-update-rule}
\end{equation}
where $\bm{\Omega}^{t}_{u}$ are the updated parameters obtained from the LSTM. Compared to Equation~(\ref{eq:optimizee-adjust}), we have changed $\bm{g}^t = \bm{\Omega}^{t}_{u}$ to $\bm{g}^{t} = \alpha \cdot  \sigma(\bm{\Omega}^{t}_{u})$ where $\alpha$ is a hyper parameter and $\sigma$ is a non-linear activation function (tanh in our implementation). This update rule applies a non-linear activation to the LSTM parameters and controls the strength of the update using $\alpha$. Informally, $\alpha$ can be interpreted as a learning rate which helps the quantum learner in adjusting the new parameters. Since there is no direct gradient information, a destructive update (i.e., $\bm{\theta}^{t} = \bm{\Omega}^{t}_u$) would make the optimizee's parameters prone to oscillation due to the noisy output from LSTM. Additionally, the raw parameter updates from LSTM are unbounded in  $\mathbb{R}$. %(-\infty,+\infty)$. % $(\mathbb{R}^+, \mathbb{R}^-)$. optim
A clipping function (e.g., tanh) that clips the values between certain maximum and minimum values, leads to more stable updates which consequently yield a smoother parameter update.

\textbf{Replay Buffer}: In earlier works, the LSTM network was able to compute a good descent direction based on gradient information provided to it during training. Even in such works as~\cite{verdon2019learning}, the LSTM network was only used to learn the initial parameters for optimizing a QNN and then the conventional gradient descent method was used. \emph{However, we consider a harder case where no gradient information is available}. It is already known that a short horizon bias problem exists for learned optimizers for classical neural networks~\cite{wu2018understanding}. In the case of optimizing quantum networks without using any gradient information, this problem becomes even harder to overcome. 

In meta-optimization a short horizon bias occurs when the meta-optimizer becomes biased to providing updates akin to taking short steps towards minima. These updates do not cause a significant overall decrease in optimizee's cost function. This effect is more pronounced when we operate in parameter space as opposed to gradient space since the LSTM network does not get any information about the curvature of the optimization surface.  Running optimization in a finite horizon window in parameter space can then cause LSTM to suggest incorrect updates to the optimizee network. In the specific case of QNNs, the parameters correspond to a rotation of given input state about a particular axis. An incorrect update can very easily cause the quantum state to be rotated incorrectly and therefore lead to an increase in the value of the objective function we're interested in minimizing. To overcome this issue, we develop a technique inspired by work in reinforcement learning~\cite{mnih2013playing}. Throughout training, we keep track of past history of parameters and their corresponding cost function values in a ``replay buffer".  At the start of training we instantiate a double ended queue dubbed as a \emph{replay-buffer} $\mathcal{B}$ of a finite capacity $R$. For meta-iteration $t=1 \dots T$ we observe a history of parameters $\bm{\theta}^t$ and the corresponding cost $C(\bm{\theta}^t)$. If $C(\bm{\theta}^{t+1}) < C(\bm{\theta}^t)$ then we add the state $s = [\bm{\theta}^{t+1}, C(\bm{\theta}^{t+1}), \Delta C(\bm{\theta}), \bm{h}^{t+1}]$ to the replay buffer. Once the meta-iteration ends and $\mathcal{L}(\bm{\Phi})$ is computed, if the QNN cost function is diverging,  we seed the parameters for the next meta-iteration by performing the following update:
\begin{equation}
    \begin{aligned}
       \bm{\theta}^{T+1} &= \tau \cdot \bm{\theta}^{T} + (1 - \tau)\cdot \bm{\theta}_{s}\\
       \tau &= \frac{\tau}{(1+ \zeta \cdot t)}.
    \end{aligned},
    \label{eq:replay_buffer_update}
\end{equation}
% \added[id=hm]{it is not clear what the subscript notation means. Can we do without it?}\added[id=ak]{The returned state from replay buffer consists of multiple elements, we need to be precise about what exactly we're using. See Algorithm~\ref{alg:rnn_opt}}
where $\zeta$ is a decay factor that adjusts the blending coefficient $\tau$ as the training progresses and $\bm{\theta}^{T}$ are the optimizee parameters at the end of a previous unrolled meta-iteration loop. $\bm{\theta}_{s}$ are the sampled parameters from the replay buffer that are chosen according to a fixed policy function $\pi(\mathcal{B})$. In our work this policy corresponds to choosing the parameters that lead to most cost decrease over the previous unroll iterations. We expect that future works in this directions will be able to learn $\pi(\mathcal{B})$ along with parameters.  We set $\tau$ to be $0.9$ to put more emphasis on the current parameters than sampled parameters. As training progresses and we observe divergent behavior, we decrease $\tau$ according to Equation~\ref{eq:replay_buffer_update} with $\zeta = 0.99$ and $t$ indicating the overall global step of optimization iteration.

% As training progresses we decrease $\tau$ to place a stronger emphasis on sampled parameters since $\mathcal{B}$ consists of the most recent history of best performing parameters at previous time steps. We empirically found $\zeta=0.9$ to provide good convergence performance on the datasets considered in our study.  
% The parameters in $s_{d}$ are then passed in-place of the last time step's parameters to the LSTM. We also consider an alternate strategy where the parameters from the last time step are blended with the parameters in $s_{d}$ using a weighted combination and the weights are updated using an annealing rule.

% \input{sections/alg_typeset}

%% Important: Either provide the algorithm or remove the second sentence.
\textbf{Overall Algorithm}: The overall flow of the algorithm is shown in Figure~\ref{fig:rnn_flow}. At a given timestep $t$, we expect the cost function value $C(\bm{\theta}^{t-2})$ and the parameters $\bm{\theta}^{t-1}$. We then compute the cost $C(\bm{\theta}^{t-1})$ using Equation~(\ref{eq:costfn}). A cost delta is computed $\Delta C(\bm{\theta}) = C(\bm{\theta}^{t-2}) - C(\bm{\theta}^{t-1})$, which is then used as a decision metric and a feature in the input. In the former role, if $\Delta C(\bm{\theta}) < 0$ then a sample (described earlier) is committed to the replay buffer $\mathcal{B}$. Then, depending on the availability of elements in $\mathcal{B}$, the parameters are either sampled or retained from previous step. An input feature is then pre-processed using Equation~(\ref{eq:input-function}). The updated parameters $\bm{\Omega}^{t}_{u}$ are then obtained from the LSTM and $\bm{\theta}^{t}$ is computed using Equation~(\ref{eq:param-update-rule}). In the initial conditions, we set $\Delta C(\bm{\theta}) = 1.0$.

\section{Theoretical Performance}
We now analyze the theoretical performance gain that our algorithm can provide in the context of parameterized quantum circuits. We note that our algorithm is not predicated on an existence of an efficient data structure like QRAM~\cite{HHL}, rather the performance gains are expected due to non-computation of gradient on quantum devices.

\begin{theorem}
Consider a  variational quantum circuit consisting of $q$ qubits, $L$ layers and $k$ single qubit gates per layer. Let $\mathcal{A}$ be an optimization algorithm running on a classical device helping the circuit minimize a cost function $C(\theta)$. Then:
\begin{itemize}[noitemsep,topsep=-\parskip,leftmargin=*]
\item If $\mathcal{A}$ is a gradient-dependent algorithm the the total time for one full pass (forward and backward) takes $O(2(qLk)^2\delta t_{f})$, where $\delta t_{f}$ is the time for a forward pass.
    
\item If $\mathcal{A}$ does not require a gradient, then the total time for one full pass takes $O(2(qLk)\delta t_{f})$.
\end{itemize}
\label{theo:run-time}
\end{theorem}
\begin{proof}
Consider an instance of gradient dependent algorithm $\mathcal{A}^{g}$. To update the parameters for the next step, it requires $\nabla_{\theta} C(\theta)$. For a circuit running on NISQ computer, currently the only method to compute gradients is to use the parameter shift method~\cite{schuld2019evaluating}. The expected gradient w.r.t single component $\theta_i$ is given as:
\begin{equation}
    \nabla_{\theta_{i}} C(x; \theta_{i}) = \bra{\psi_{x}}\nabla_{\theta_{i}}\mathcal{F}_{\theta_{i}}(\bm{\hat{O}})\ket{\psi_{x}},
    \label{eq:gradient-def}
\end{equation}
where $\ket{\psi_{x}} = \bm{U}_{0}(x)\ket{0}$. The quantity $\nabla_{\theta_{i}}\mathcal{F}_{\theta_{i}}(\bm{\hat{O}})$ is the gradient w.r.t to the $i^{th}$ component of the parameter vector. The parameter shift rule states: 
\begin{equation}
    \begin{aligned}
       \nabla_{\theta_{i}}\mathcal{F}_{\theta_{i}}(\bm{\hat{O}}) &= c[\mathcal{F}_{\theta_{i} +s}(\bm{\hat{O}}) - \mathcal{F}_{\theta_i - s}(\hat{O}))] \\
       \mathcal{F}_{\theta_i}(\bm{\hat{O}}) &= \bm{U}^{\dagger}(\theta_{i}) \bm{\hat{O}} \bm{U}(\theta_{i}).
    \end{aligned},
    \label{eq:gradient-update}
\end{equation}
where $c$ is a scaling constant ($c = 0.5$) and s is a constant shift angle about which the gradient is computed. To successfully evaluate $\nabla_{\theta_{i}} C(\theta)$ one needs two evaluations of the quantum circuit with shifted parameters. Consequently, the time for one gradient computation over the entire quantum circuit is $O(2qlK)$. Then, for a full pass (i.e. computing the cost and gradient), the computation time for $\mathcal{A}^{g}$ scales as $O(2(qlK)^2\delta t_{f})$. In contrast, a gradient free method (like ours), does not incur the penalty of computing the gradient and thus the overall computation time scales as $O(2(qlk)\delta t_f)$. 
\end{proof}

The reduction of a quadratic run-time to a linear run-time is significant since this will allow a QNN to scale to a larger number of data points. Although, our method does incur a storage overhead of $O(R)$, for all practical scenarios $O(R) \ll  O(qlK)$.

% \begin{figure*}[t]
%     \centering
%     \begin{subfigure}[b]
%     \centering
%      \includegraphics[width=\textwidth]{figures/iris_randn_q_4_pure_ckt.png}
%     \caption{Iris}
%     \end{subfigure}
%     \begin{subfigure}[b]{.3\textwidth}
%     \centering
%      \includegraphics[width=\textwidth]{figures/gaussian_randn_q_4_pure_ckt.png}
%     \caption{Gaussian}
%     \end{subfigure}
% \end{figure*}

\begin{figure*}[t!]
    \centering
    \begin{subfigure}{.25\textwidth}
    \centering
     \includegraphics[width=\linewidth]{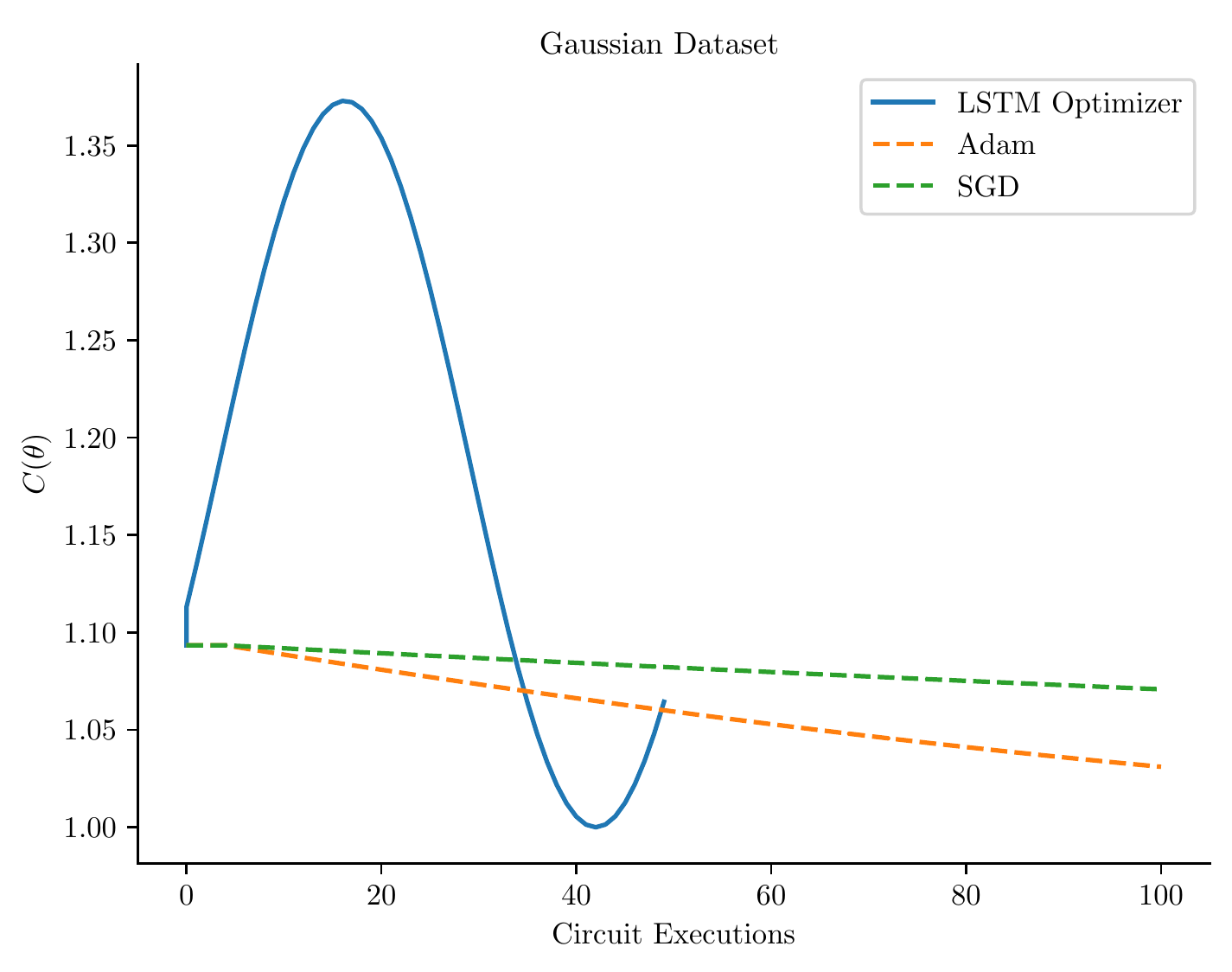}
     \caption{}
     \label{fig:gaussian_no_rb}
    \end{subfigure}\hfil
    \begin{subfigure}{.25\textwidth}
    \centering
     \includegraphics[width=\linewidth]{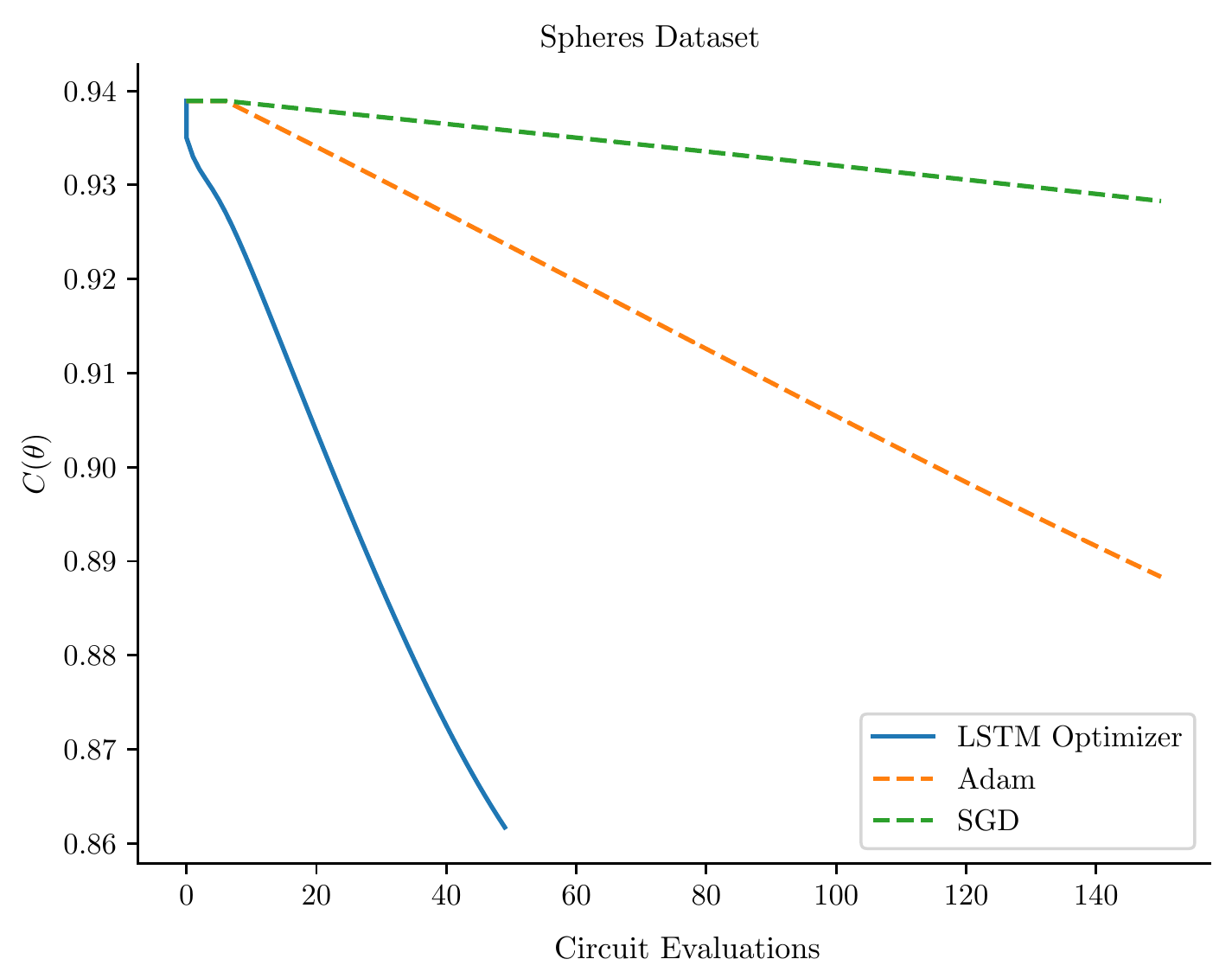}
     \caption{}
     \label{fig:spheres_no_rb}
    \end{subfigure}\hfil
    \begin{subfigure}{.3\textwidth}
    \centering
     \includegraphics[width=\linewidth]{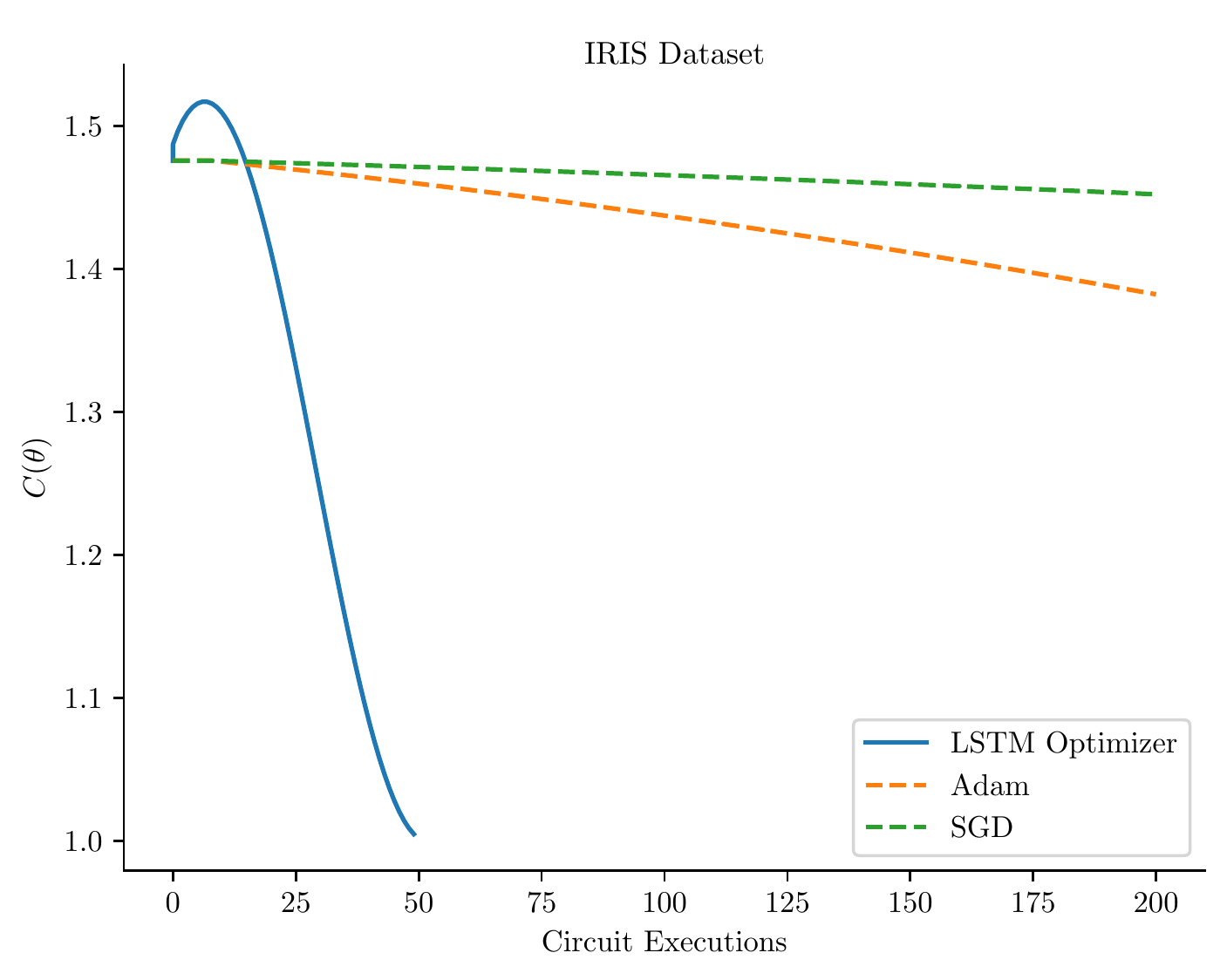}
     \caption{}
     \label{fig:iris_no_rb}
    \end{subfigure}

    \medskip
    \begin{subfigure}{0.25\textwidth}
        \includegraphics[width=\linewidth]{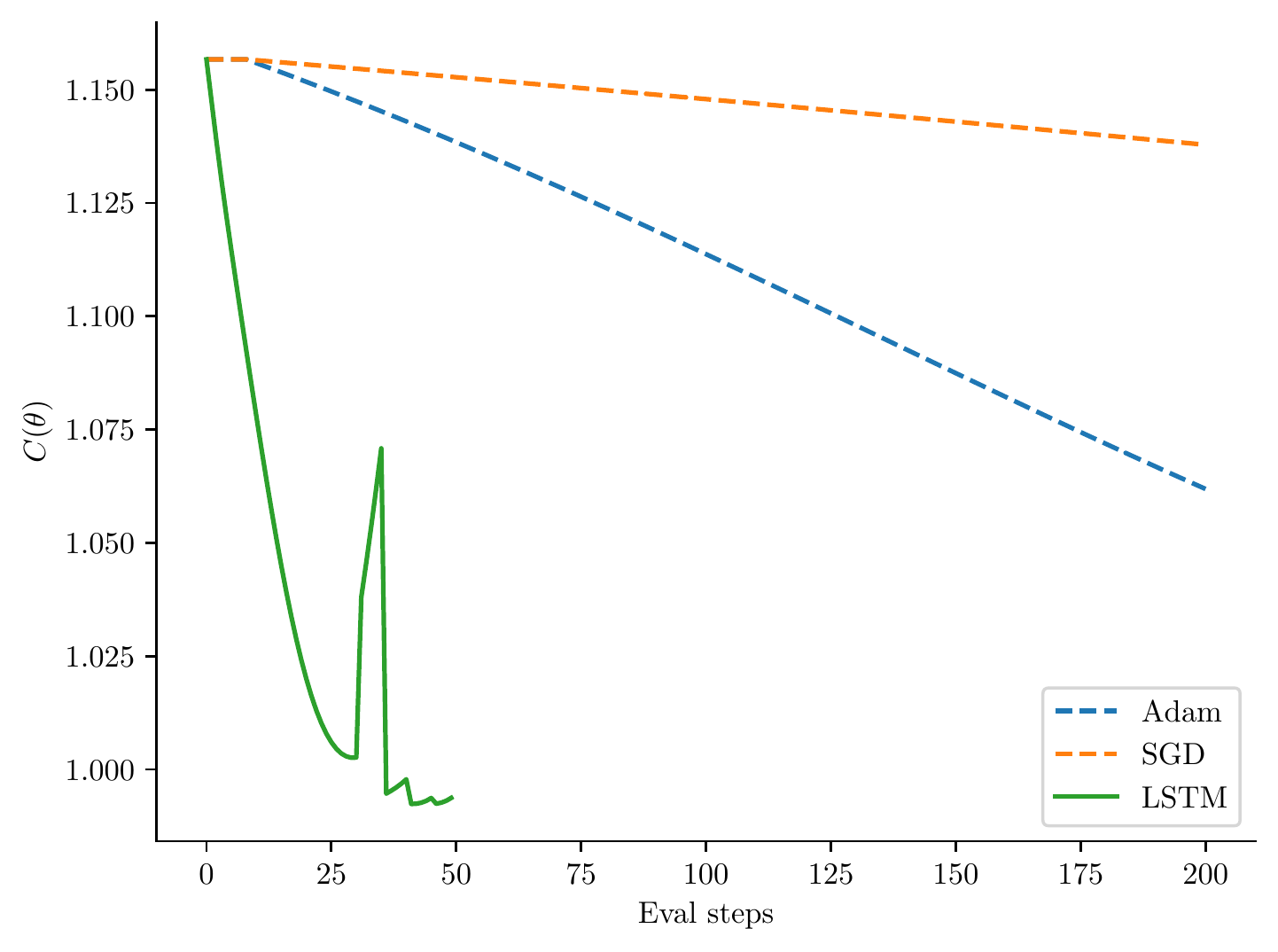}
        \caption{Gaussian Clusters}
        \label{fig:gaussian_rb}
    \end{subfigure}\hfil % <-- added
    \begin{subfigure}{0.25\textwidth}
        \includegraphics[width=\linewidth]{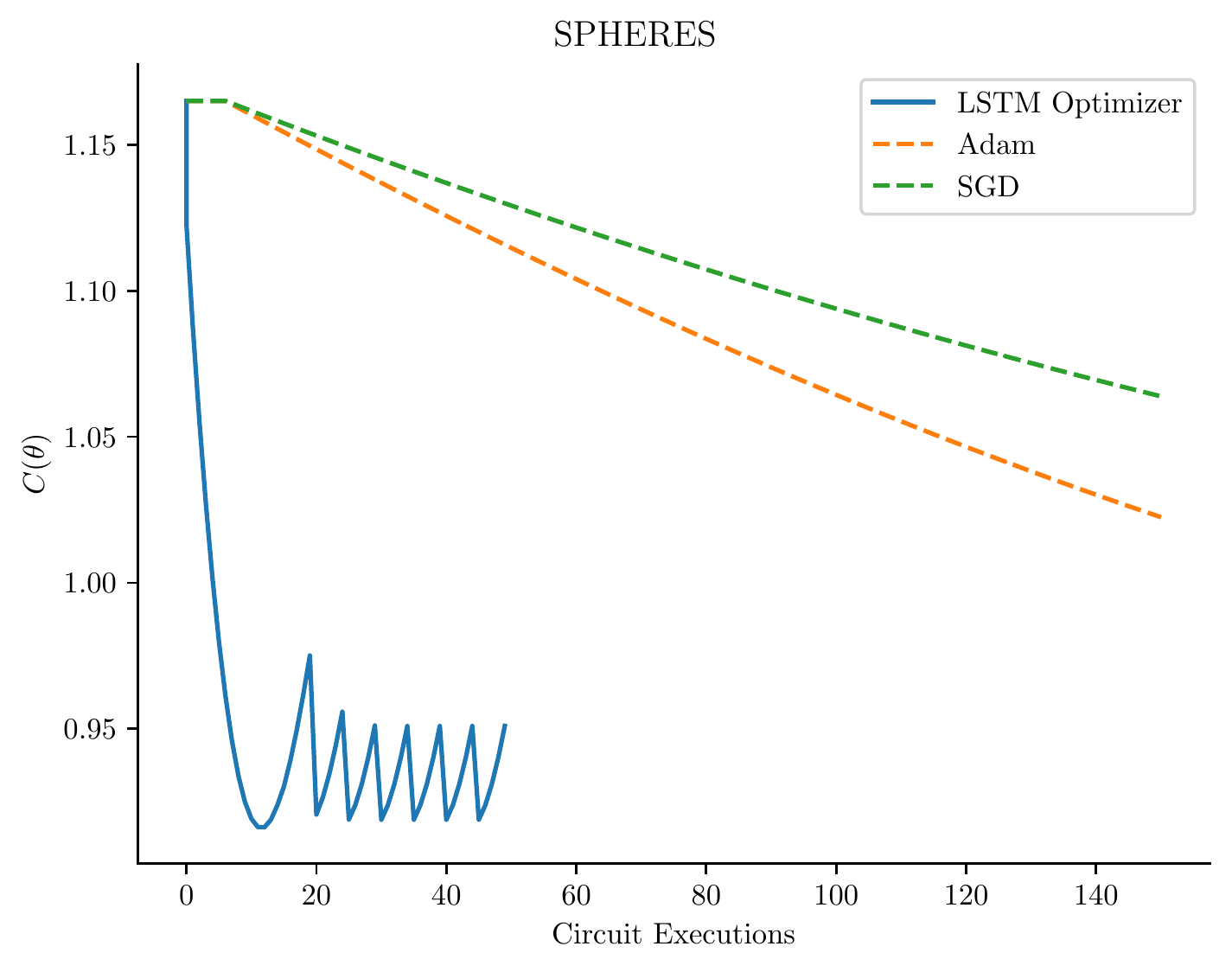}
        \caption{Spheres 3d}
        \label{fig:spheres_rb}
    \end{subfigure}\hfil % <-- added
    \begin{subfigure}{0.25\textwidth}
        \includegraphics[width=\linewidth]{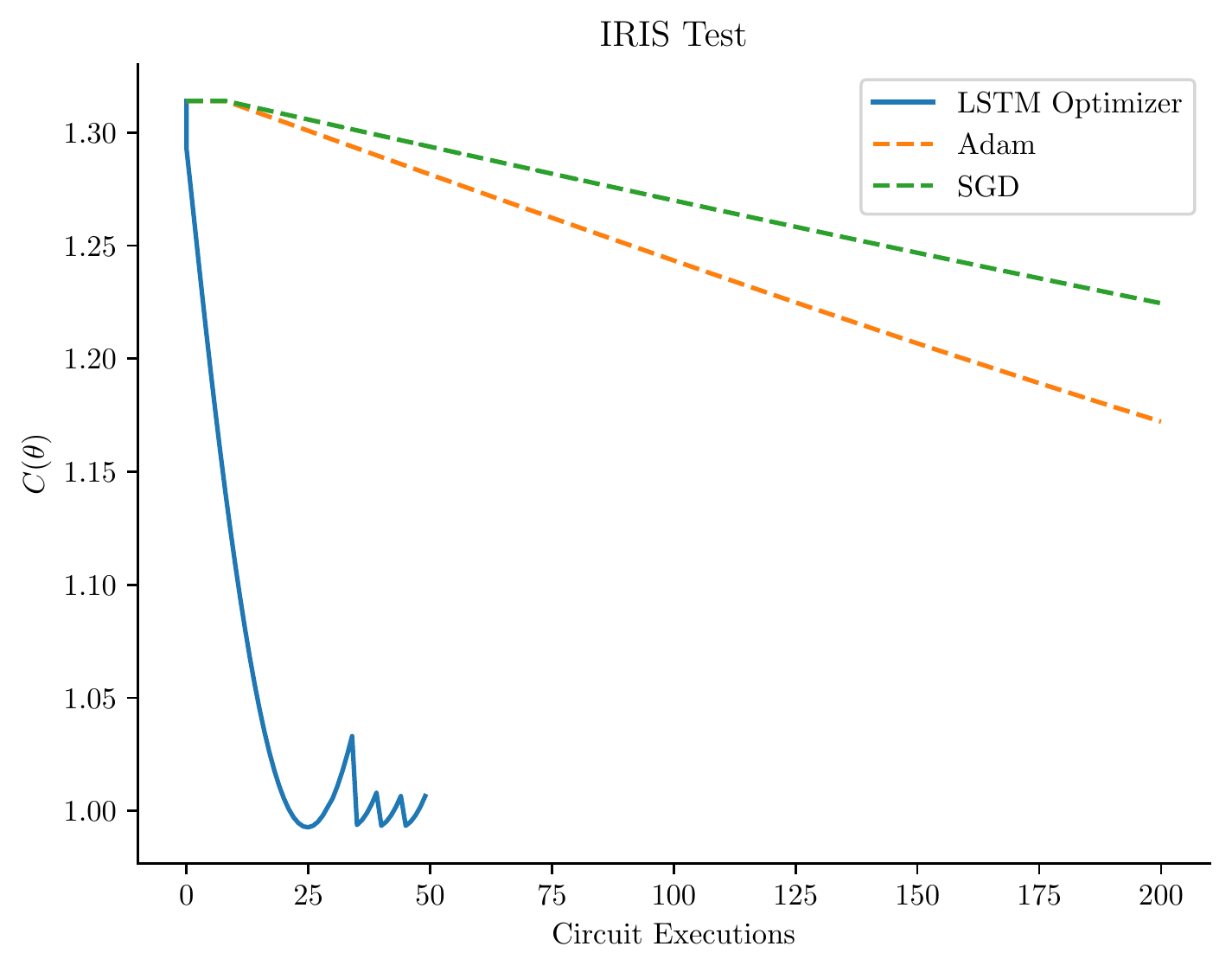}
        \caption{Iris}
        \label{fig:iris_rb}
    \end{subfigure}
    
    \caption{Performance of LSTM-based meta-optimizer on the binary classification task for three datasets. All parameters are initialized using the normal distribution. The top row shows the performance of the algorithms without any replay buffer sampling and the bottom row shows the performance with replay buffer sampling. In both cases, the LSTM based optimizer is able to achieve a lower cost in significantly fewer circuit evaluations.}
    \label{fig:perf_qckt_ceval}
\end{figure*}

\section{Numerical Experiments}\label{sec:expts}

In this section we show numerical experiments that demonstrate the effectiveness of our method. The goal of our experiments is twofold. First, we wish to empiricially demonstrate that using a LSTM based meta-optimizer can lead to better convergence with lower number of circuit executions than existing gradient based or gradient free methods. Second, we wish to show that our method is more useful than gradient based methods in realistic scenarios where the number of shots on a quantum device are limited. In this latter case we demonstrate that the LSTM based optimizer is able to evolve a significantly better optimization trajectory than a gradient based algorithm. %

\subsection{Experiments on Machine Learning Datasets}
\begin{figure}
    \centering
    \includegraphics[width=\columnwidth]{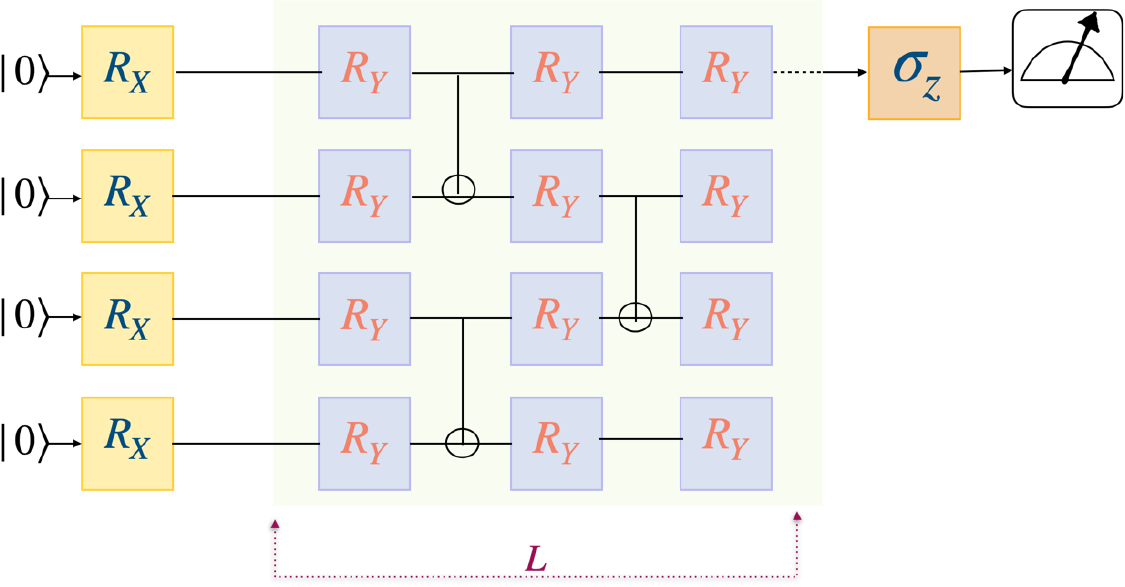}
    \caption{The quantum circuit used in our study. The $R_{Y}$ gates are parameterized by variational parameters. The $R_{X}$ gates
    applied to each qubit constitute $\bm{U}_0$ in Equation~\ref{eq:qnn}}.
    \label{fig:qckt_expt}
\end{figure}

We first present our experiments on a pattern classification task using an instance of a layered ansatz as shown in Figure~\ref{fig:qckt_expt}. The first layer in the ansatz embeds the input data $\bm{x}$ into the corresponding angles of a RX gate, where $RX(x_j)= e^{-ix_j \sigma_{x}/2}$ and $\sigma_{x}$ is the Pauli-X matrix. The subsequent layers apply a parameterized RY rotation and are entangled using the CZ gate. Since the number of input qubits is equal to the dimensionality of data, we consider three datasets with increasing dimensionality - Gaussian, Spheres and Iris. The Gaussian dataset is a synthetic dataset where two-dimensional clusters are instantiated by drawing samples from a multivariate Gaussian distribution of given parameters. Similarly, the spheres dataset is a collection of 3-d points which form concentric spheres with one sphere enveloping the other. Finally, we consider a truncated Iris dataset that consists of only two classes and four features that describe a particular species. In all these datasets, the labels are encoded into $\{-1, +1\}$ and the optimization task is to find variational parameters that minimize Equation~\ref{eq:costfn} with Pauli-Z gate being the observable. 

In the experiments we do not apply any pre-processing or post processing on the input data. We benchmark our LSTM optimizer against two commonly used gradient based algorithms - ADAM~\cite{kingma2014adam} and Gradient Descent. For the LSTM optimizer, the ansatz is run for 50 iterations while for gradient based algorithms we run it for 25 iterations. We then profile the cost obtained by the corresponding algorithms with the number of circuit evaluations. The learning rate for gradient based algorithms is set as $1e^{-2}$ and $\alpha = 0.1$ for the LSTM optimizer.

Figure~\ref{fig:perf_qckt_ceval} shows the results of our experiments on the three datasets. The top (bottom) rows show the performance without(with) replay buffer sampling. In all cases, the LSTM optimizer is able to find parameters that minimize the cost function in far fewer circuit evaluations than gradient based optimizers. This result is an empirical validation of Theorem~\ref{theo:run-time} and intuitively makes sense since parameter shift rules evaluate the same circuit twice per parameter component to estimate the gradient at a given timestep.

Figure~\ref{fig:gaussian_no_rb} shows a phenomenon unique to LSTM based optimizers. After finding a good descent direction, the LSTM based optimizers tend to over optimize and thus oscillate about some fixed point. This ``sinusoidal" osciallation is hard to control in optimization problems since it's not easy to guess the stationary points in multivariate objective functions. Our replay buffer sampling strategy is aimed to control this phenomenon in a more principled manner and it's effects are visible in Figure~\ref{fig:gaussian_rb} where the parameter mixing and decay lead to a more stabler convergence. The effects are also visible in Figures~\ref{fig:spheres_rb} and Figures~\ref{fig:iris_rb} where the cost function does not diverge after reaching a minimum point. 

\begin{figure*}[t]
    \centering
    \begin{subfigure}{0.25\textwidth}
        \includegraphics[width=\linewidth]{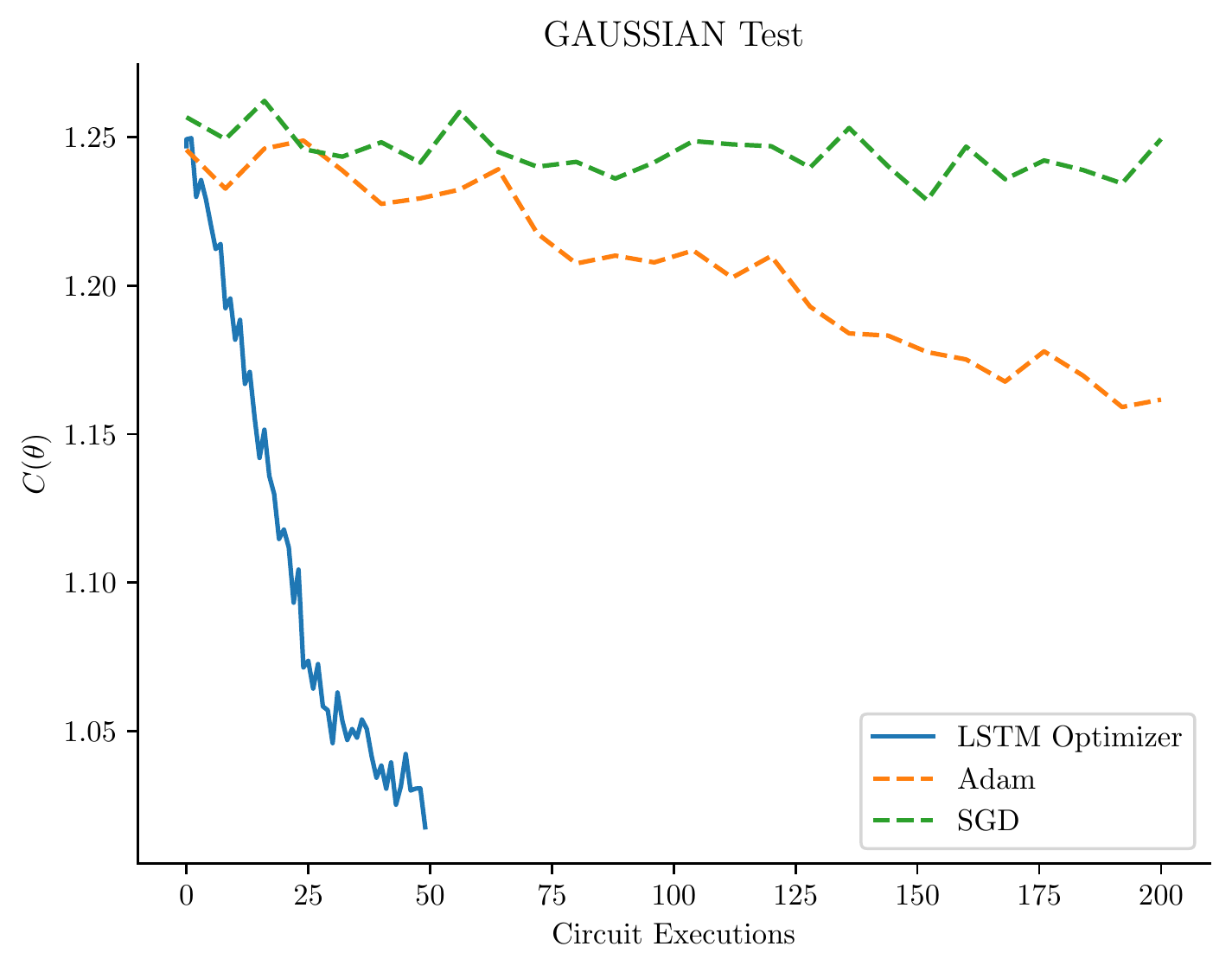}
        \caption{Gaussian Clusters}
        \label{fig:1}
    \end{subfigure}\hfil % <-- added
    \begin{subfigure}{0.25\textwidth}
        \includegraphics[width=\linewidth]{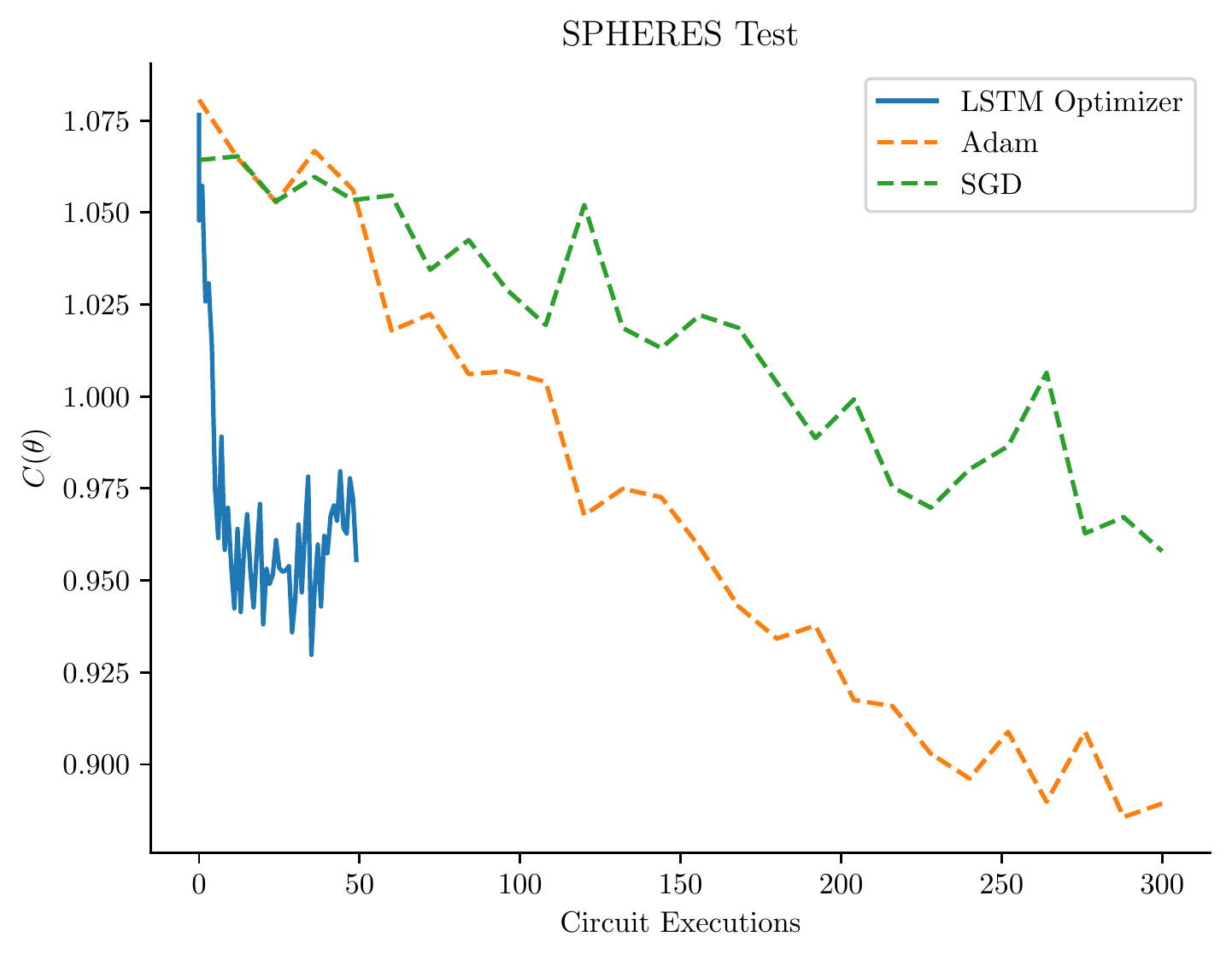}
        \caption{Spheres 3d}
        \label{fig:2}
    \end{subfigure}\hfil % <-- added
    \begin{subfigure}{0.25\textwidth}
        \includegraphics[width=\linewidth]{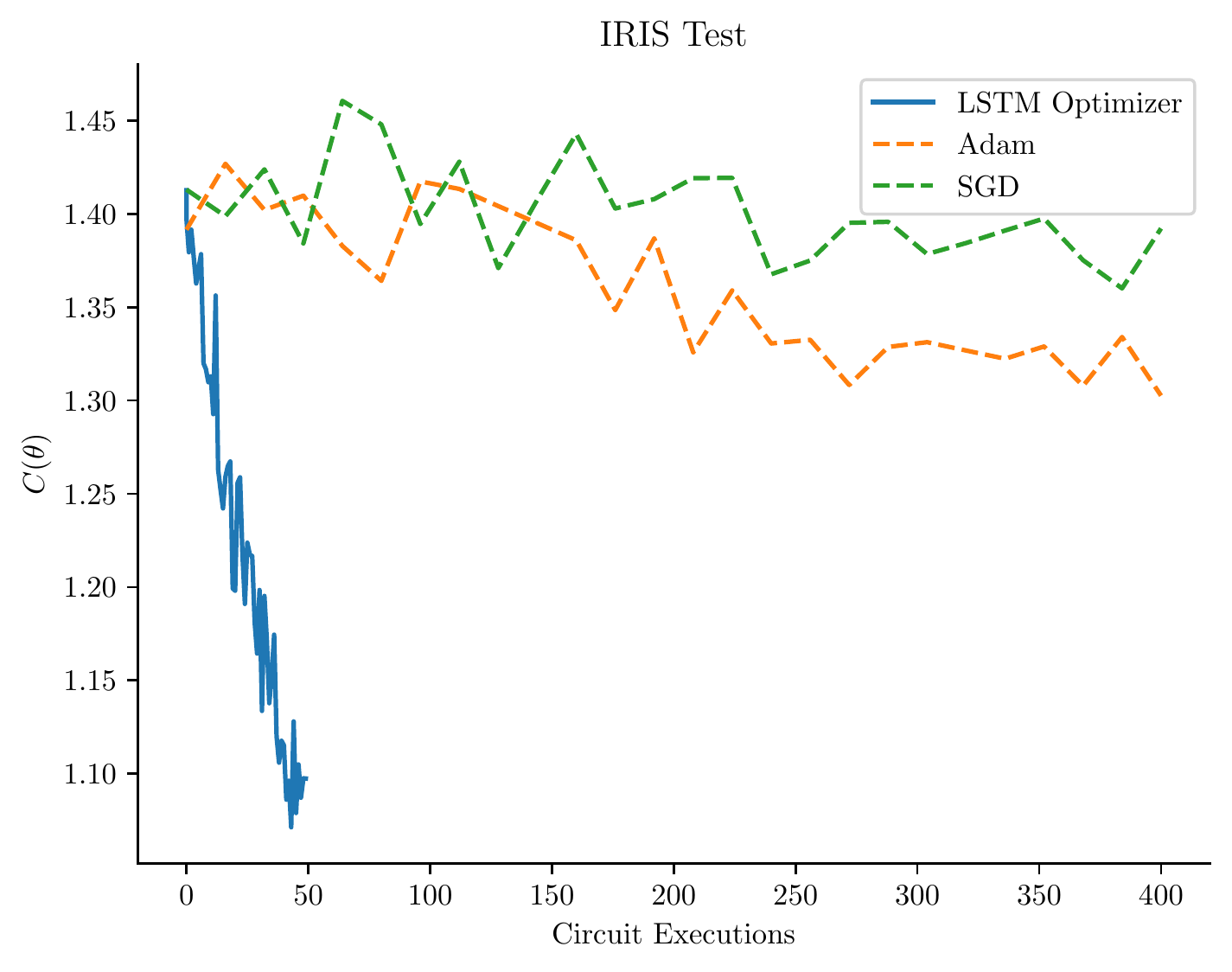}
        \caption{Iris Dataset}
        \label{fig:3}
    \end{subfigure}
    \caption{Results on experiments with limited number of shots. In the limited shot setting, the LSTM based optimizer is able to minimize the cost function more effectively than gradient based approaches.}
    \label{fig:qckt_eval_lim_shots}
\end{figure*}

%% performing experiments with 1000 shots as well
\subsection{Experiments with Limited Shots}
In the experiments on the machine learning datasets we simulated the performance of the optimization algorithms assuming that number of available shots on the device were infinite. This assumption is however not realistic since most NISQ devices have limited number of shots for measurement. Hence, we measure the performance of the LSTM optimizer in the setting where the number of shots is highly limited.   

%% TODO: Rewrite this.
Figure~\ref{fig:qckt_eval_lim_shots} shows the results of our experiments on a simulated quantum device with only 100 shots for measurement. We \emph{a-priori} expect the stochasticity in the sampling to manifest in the optimization process and affect the overall quality of minima for both gradient free and gradient based optimizers. In the former case, the noisy measurement of cost functions with given parameter values make it hard for the LSTM to suggesst ``good" updates and in the latter case the noisy measurements of cost functions are coupled with noisy estimation of gradient. The noisy estimation of gradient in this case is not equivalent to stochastic averaging over mini-batches in classical deep learning algorithms and hence does not guarantee a better convergence rate. In the figure, we see that the LSTM based optimizer is able to obtain a lower minima in atleast two datasets and is able to find a minima faster than gradient based algorithms in all datasets. This demonstrates that stochastic measurement of gradient can lead to slower and sub-optimal coverage. We further demonstrate the performance of the QNN on these datasets with 1000 shots in the appendix. To conclude, in a limited shot setting a gradient free algorithm like ours can be useful as a drop in replacement for a gradient based optimization algorithm.

\subsection{Comparison Against Other Gradient Free Approaches}
We now benchmark our algorithm against another well known gradient-free algorithm - Simultaneous Perturbation Stochastic Approximation (SPSA) algorithm~\cite{robbins1951stochastic} which approximates gradients by taking finite differences between perturbed parameter vectors. The perturbations are generated stochastically. 

We perform experiments with a four qubit, five layer variational circuit whose minimum parameters correspond to the minima of the following cost function:

\begin{equation}
    C(\bm{\theta}) = \bigotimes_{i=1}^{n} \bra{\psi(\bm{\theta})}\sigma_{z}(i)\ket{\psi(\bm{\theta})}
    \label{eq:cost_fn_spsa}
\end{equation}

\begin{figure}[h]
    \centering
    \includegraphics[width=\columnwidth]{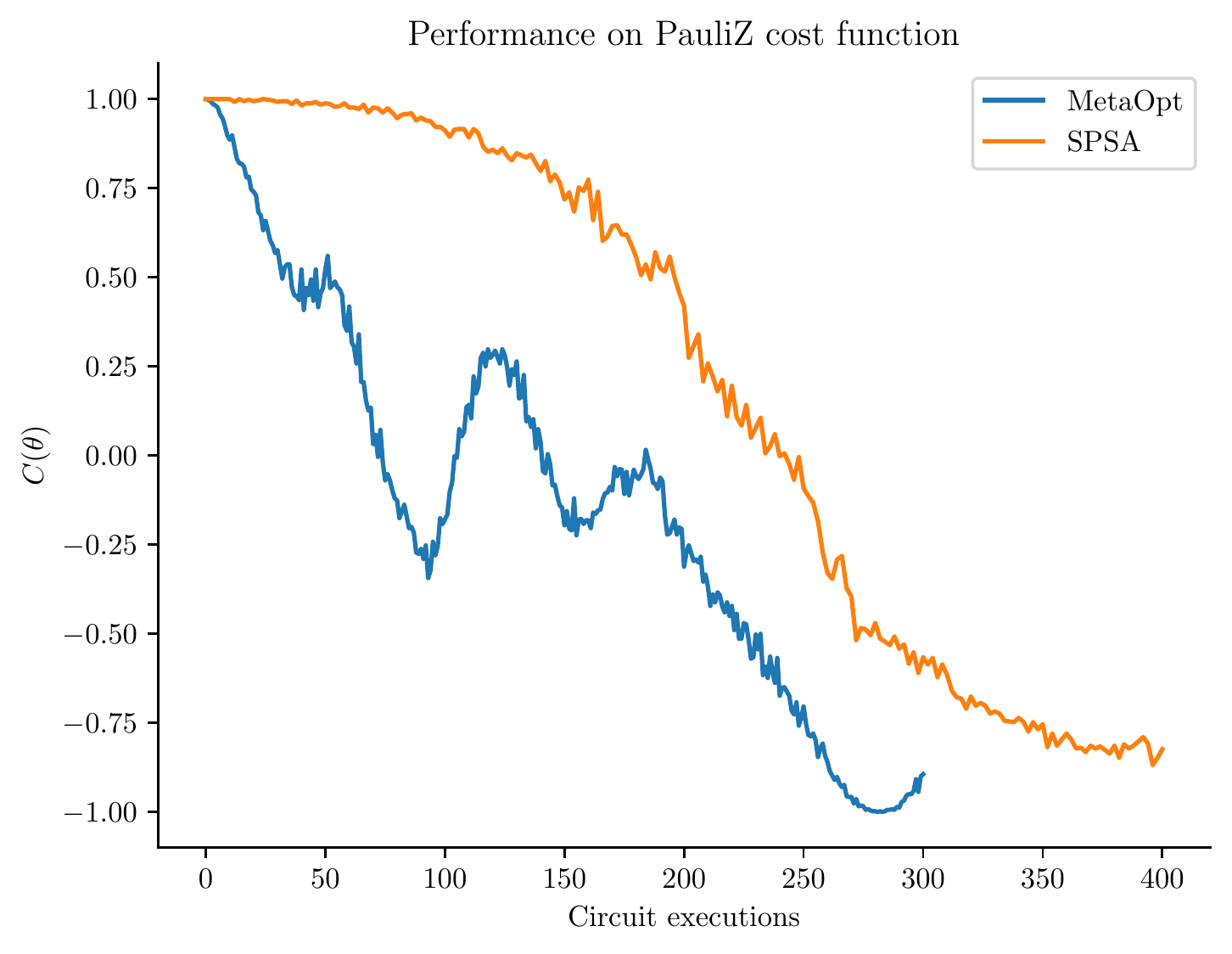}
    \caption{Performance of LSTM based optimizer benchmarked against the SPSA algorithm. For the same cost function, SPSA makes an extra call to the quantum circuit which results in more number of circuit evaluations. LSTM based optimizer outperforms SPSA in terms of cost function minimum.}
    \label{fig:spsa_vs_lstm}
\end{figure}

Where $\ket{\psi(\bm{\theta})}$ is a variational circuit consisting of strongly entangling layers i.e. rotation gates with all-to-all entanglement and $\sigma_{z}(i)$ is the Pauli-Z matrix acting on the $i^{th}$ qubit as the observable. The results of our experiments are shown in Figure~\ref{fig:spsa_vs_lstm}. We can see that the LSTM Optimizer is able to find a better quality minima than SPSA as well. Moreover, SPSA makes an extra call to the circuit per optimization step making it slightly more expensive than LSTM optimizer. We conclude that the our gradient free algorithm can help VQCs scale to larger and more complex problem instances in the future.

\subsection{Time Profiling Results}

\begin{figure}[h]
    \centering
    \includegraphics[width=\columnwidth]{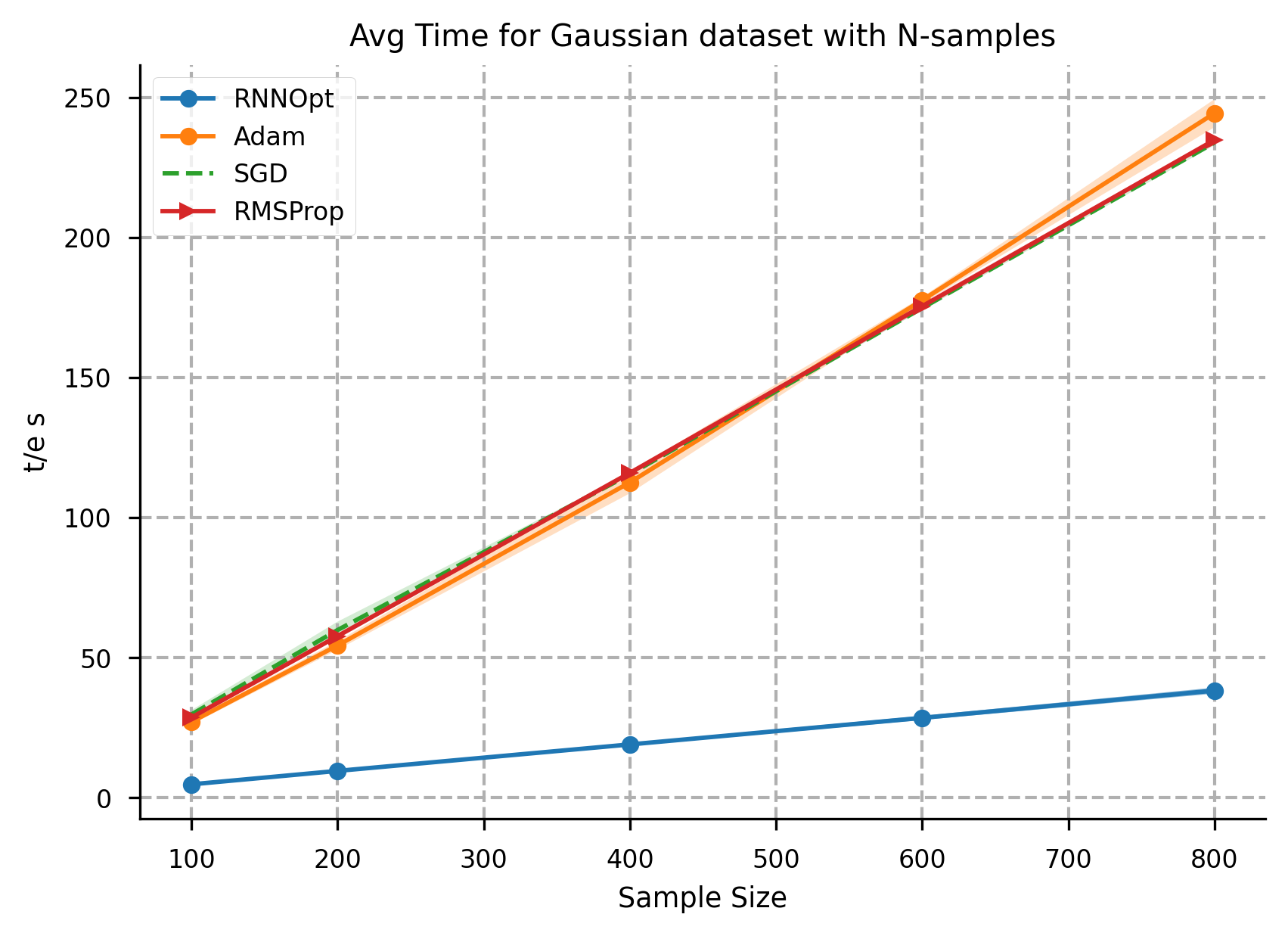}
    \caption{Time profiling results for different sizes of Gaussian datasets.}
    \label{fig:gauss_expt_results}
\end{figure}

In our earlier experiments we noticed that LSTM optimizer converges to a minima in fewer circuit evaluations than other methods. This implied that our method could be used to scale to data with larger number of points. To test this hypothesis, we perform experiments for a pattern classification task by generating $N$ samples from a Gaussian distribution of finite mean and covariance. We benchmarked the time per epoch (i.e. time it takes to go through an entire dataset) for our method against Adam, Gradient Descent Optimizer and RMSProp~\cite{tieleman2012lecture}. 

The results of our experiment are shown in Figure~\ref{fig:gauss_expt_results}. We vary the sample size $N = \{100, 200, 400, 600, 800 \}$ and measure the scaling of the time per epoch it takes for the optimizers. It is clear from the figure that the LSTM optimizer is able to scale to larger data instances without a significant increase in the time per epoch. In fact, our method is nearly \emph{three} times faster than any gradient based method on the same dataset. This result is a positive indication that development of novel gradient free methods is a novel research direction and improvements may lead to VQCs being able to scale to handle data at the scale which is currently being handled by classical algorithms. 

\section{Related Work}
Meta optimization is an actively researched area in the field of deep learning. Many different meta-optimization algorithms have been proposed. Notably,~\cite{andrychowicz2016learning} propose an LSTM based meta-optimizer that accepts the cost function gradient and the previous time parameters as inputs to a LSTM network and recover new parameters in the form of LSTM output while Ravi~\emph{et al.}~\cite{ravi2017optimization} embed parameters of neural network into the cell state of the LSTM and the gradient as the state of the candidate cell state. The output is then a non-linear combination of these two embeddings in the LSTMcell.  A reinforcement learning perspective is considered by~\cite{li2016learning, li2017efficient} where a guided policy search algorithm is used to find new parameters given information from previous time steps (over a finite horizon). Another direction in meta-optimization is considered by~\cite{bello2017neural} in the form of neural architecture search over the input space formed by the symbols in various update rules. They evaluate the viability of different optimization algorithms resulting from different combinations of the symbols and analyze the efficacy of the best performing combination. 
In the quantum case, there has been sparse but consistent work towards meta-optimization as well. The motivation for these works is towards exploring solutions to problems in quantum chemistry (e.g. finding the lowest energy Hamiltonian for a given system) and graph optimization (e.g. finding the parameters corresponding to the max-cut problem in QAOA). In~\cite{verdon2019learning}, the authors propose a strategy of learning the best initial parameters using the LSTM setup of~\cite{andrychowicz2016learning}. These initial parameters are then used in a quantum circuit to avoid barren plateaus and are optimized using conventional gradient based algorithms. Wilson~\emph{et al.}~\cite{wilson2021optimizing} propose a similar setup to~\cite{andrychowicz2016learning} and utilize similar inputs to the LSTM based meta-optimizer. Their assumption is based on the cheapness of gradient evaluation on the quantum circuit, which unfortunately for the NISQ devices is not the case. Another work by~\cite{khairy2020learning} proposes a reinforcement learning based approach in a Quantum Approximate Optimization Algorithm (QAOA) setting by trying to learn a policy that suggests optimal parameters for a max-cut problem with two clusters and SVM/SVR for predicting the circuit depth~\cite{shaydulin2021classical}. Their proposed method leverages PPO~\cite{schulman2017proximal} algorithm for estimating the optimal policy. Our work can also be interpreted as a variant of reinforcement learning method where we learn the policy directly instead of just suggesting the mean of a Gaussian policy. Second, our method makes no presumption on the \emph{size} of the action space unlike this work where the action space is limited to just two parameters. To the best of our knowledge, our work is the first to consider a gradient free meta-optimization algorithm with significantly different inputs and update strategies. Table~\ref{tab:relwrk} summarizes the existing meta-optimization methods in quantum computing and highlights the differences from our method.

\begin{table}[h]
    \centering
    \tiny
    \begin{tabular}{p{0.11\textwidth}|p{.09\textwidth}|p{.09\textwidth}|p{.09\textwidth}}
    \toprule
    Method & Quantum Gradient Computation & Replay Buffer & Adaptive Parameter Updates \\
    \midrule
    %% TODO: change the method name better
    Meta-optimization w/ Gradients~\cite{wilson2021optimizing}& \ding{51} & \ding{55} & \ding{55} \\  
    \hline
    Learning initial parameters w/ LSTMs~\cite{verdon2019learning} & \ding{51} & \ding{55} & \ding{55} \\ 
    \hline
    Ours & \ding{55} & \ding{51} & \ding{51}
    \end{tabular}
    \caption{Algorithmic differences between quantum meta-optimization algorithms proposed in literature.}
    \label{tab:relwrk}
\end{table}

\section{Conclusion}
We have proposed a gradient free meta optimization algorithm for quantum neural networks that can potentially scale up to larger dataset sizes in a reasonable amount of time. Experiments on different datasets show that our algorithm has a comparable quality to classical optimization algorithms while significantly outperforming them in terms of computation time. We believe that our method can be useful in exploring the answer to several open questions in the theory of variational quantum algorithms. For instance, the barren plateau problem~\cite{mcclean2018barren} is a notorious problem that occurs in randomly parametrized circuits even when their depth is shallow. It has been shown that initializing distributions of parameters play a key role in preventing barren plateaus~\cite{kulshrestha2022beinit}. It would be interesting to explore if meta-optimizers can suggest parameters that prevent occurrence of barren plateaus. In a future work, we would also like to study the generalization performance of QNNs when trained by meta-optimizers.

\section{DATA AVAILABILITY}

Data and code are available upon reasonable request from the authors.

% \section{AUTHOR CONTRIBUTIONS}

% The project was conceived by Ankit Kulshrestha and Xiaoyuan Liu. The latter provided guidance on experiments and evaluation of results. The authors thank Hayato Ushijima-Mwesigwa for his contribution towards the parameter mixing rule and the annealing hyperparameter suggestions. The authors would also like to thank Ilya Safro for his contributions in the overall suggestion of experiments and his invaluable comments on the overall approach. We also thank him for his specific suggestion of using a Spheres3d dataset as a benchmark for our algorithm.

% \section{COMPETING INTERESTS}
% The Authors declare no Competing Financial or Non-Financial Interests.

\bibliographystyle{IEEEtran}
\bibliography{ref}

%%%%%%%%%%%%%%%%%%%%%%%%%%%%%%%%%%%%%%%%%%%%%%%%%%%%%%%%%%%%%%%%%%%%%%%%%%%%%%%
%%%%%%%%%%%%%%%%%%%%%%%%%%%%%%%%%%%%%%%%%%%%%%%%%%%%%%%%%%%%%%%%%%%%%%%%%%%%%%%
% APPENDIX
%%%%%%%%%%%%%%%%%%%%%%%%%%%%%%%%%%%%%%%%%%%%%%%%%%%%%%%%%%%%%%%%%%%%%%%%%%%%%%%
%%%%%%%%%%%%%%%%%%%%%%%%%%%%%%%%%%%%%%%%%%%%%%%%%%%%%%%%%%%%%%%%%%%%%%%%%%%%%%%
\newpage
% % \appendix
% % \onecolumn

% \input{sections/supplementary.tex}

% You can have as much text here as you want. The main body must be at most $8$ pages long.
% For the final version, one more page can be added.
% If you want, you can use an appendix like this one, even using the one-column format.
% %%%%%%%%%%%%%%%%%%%%%%%%%%%%%%%%%%%%%%%%%%%%%%%%%%%%%%%%%%%%%%%%%%%%%%%%%%%%%%%
%%%%%%%%%%%%%%%%%%%%%%%%%%%%%%%%%%%%%%%%%%%%%%%%%%%%%%%%%%%%%%%%%%%%%%%%%%%%%%%

\end{document}